\documentclass[letterpaper, 10 pt, conference]{ieeeconf}  

\IEEEoverridecommandlockouts                              

\usepackage{amsmath,amsfonts}
\usepackage{algorithmic}
\usepackage{algorithm}
\usepackage{array}
\usepackage[caption=false,font=footnotesize,labelfont=rm,textfont=rm]{subfig}
\usepackage{textcomp}
\usepackage{stfloats}
\usepackage{url}
\usepackage{verbatim}
\usepackage{graphicx}
\usepackage{cite}
\usepackage{siunitx}
\usepackage{amssymb}
\usepackage{stfloats}
\usepackage{cases}

\newtheorem{theorem}{Theorem}
\newtheorem{lemma}{Lemma}

\newtheorem{definition}{Definition}
\newtheorem{remark}{Remark}
\newtheorem{assumption}{Assumption}

\title{\LARGE \bf
Deception in Reach-Avoid Game with Unknown Heterogeneous Attackers Speed Information
}

\author{Xiangkai Wu$^{1}$, Shaolin Tan$^{2}$, Wei Wang$^{3}$ and Zhen Han$^{4}$ 
\thanks{*This work was supported by National Natural Science Foundation of China under Grant 62373019, 
and Fundamental Research Fund for the Central Universities under Grant 501XSKC2025103001.
}
\thanks{$^{1}$Xiangkai Wu and Wei Wang are with the School of Automation Science and Electrical Engineering, Beihang University, Beijing 100191, China.
        {\tt\small xiangkaiwu@buaa.edu.cn; w.wang@buaa.edu.cn}}%
\thanks{$^{2}$Shaolin Tan is with the Zhongguancun Laboratory, Beijing 100094, China.
        {\tt\small shaolintan@hnu.edu.cn}}%
\thanks{$^{3}$Wei Wang is also with the Hangzhou Innovation Institute, Beihang University, Hangzhou 310051, China. (Corresponding author)}
\thanks{$^{4}$Zhen Han is with China North Vehicle Research Institute, Beijing 100072, China. {\tt\small hanzhen2018@buaa.edu.cn}}}

\begin{document}

\maketitle
\thispagestyle{empty}
\pagestyle{empty}

\begin{abstract}

This letter investigates a reach-avoid game involving two Attackers and one Defender, where the Attackers aim to maximize the number reaching the target region while the Defender seeks to minimize it. In contrast to conventional complete information formulations, we consider an information asymmetry scenario where the Attackers' heterogeneous maximum speeds are privately known but publicly disclosed to lie within continuous ranges. Existing studies on uncertain speeds, however, have primarily focused on homogeneous settings, whereas heterogeneity extends the uncertainty from a common capability level to the relative capability configuration of the Attackers. To address the resulting capture-order ambiguity over infinitely many possible speed combinations, we establish a critical speed pair framework that characterizes when different capability configurations induce different optimal capture orders, and enables the analysis of the Defender's guessing behavior and the design of information-limiting strategies for the Attackers. We demonstrate that under certain initial conditions, the Attackers can mislead the Defender into making suboptimal decisions through a slow-speed deception strategy, achieving superior payoffs compared to the complete information game. Numerical visualizations reveal the widespread occurrence of such dilemma conditions.

\end{abstract}

\section{Introduction}\label{sec:introduction}

Reach-avoid game, also known as the target-defending game or perimeter defense game \cite{shishika2020review,yan2023multiplayer}, is a differential game where the attackers attempt to reach a target area without being captured while defenders seek to intercept them before arrival. This problem has significant applications in ground and air combat \cite{deng2023reach,garcia2020multiple}, and autonomous driving \cite{exarchos2015suicidal}, attracting considerable research attention.

Existing studies on reach-avoid game variants \cite{liang2020analysis,yan2024multiplayer,shishika2018local,fu2023justification,lee2024solutions} commonly conclude that all players move at maximum capacity toward specific fixed points to optimize payoffs, leveraging positional or performance advantages. This conclusion relies on the assumption of complete information, where state information, payoff functions, and dynamic parameters are common knowledge. However, incomplete information on state \cite{shishika2021partial}, payoffs \cite[Chapter 5.1]{oyler2016contributions}, or parameters \cite{nath2022two} is more prevalent in practice.

Unlike complete information differential games, incomplete information variants lack mature methods and exhibit unique phenomena such as deception and inference \cite{dragan2015deceptive}. 
Against this backdrop, leveraging information advantages via strategic deception to offset positional or performance disadvantages constitutes a challenging yet valuable research direction.
However, exploration into how information advantages of key parameters influence game outcomes in differential games remains limited.
Reference \cite{shishika2024deception} pioneers the investigation of this issue by considering attackers with binary maximum speeds ($v_{\text{slow}}$ or $v_{\text{fast}}$) in target defending games, demonstrating that attackers can induce the defender to make costly misjudgments via an information-limiting strategy. References \cite{wu2025deception} and \cite{wu2026inducing} extend this work by adopting a critical speed framework where the defender only knows that the attackers' maximum speed exceeds a known lower bound.

However, previous studies \cite{shishika2024deception,wu2025deception} investigate only one-dimensional rotating turret defender, neglecting the simple-motion defender whose two-dimensional strategy space complicates attackers' exploitation of information advantages. In addition, works addressing simple motion defender \cite{garcia2019strategies,yan2021cooperative,deng2023reach} rarely incorporate information advantages.
Moreover, the works \cite{shishika2024deception,wu2025deception,wu2026inducing} assume attackers with homogeneous capabilities, which oversimplifies practical heterogeneous scenarios and leaves the strategic impact of relative capability configurations unexplored.

This letter addresses these limitations by considering a two-Attackers one-Defender (2A1D) reach-avoid game where Attackers have heterogeneous, unknown maximum speeds within continuous ranges. 
The main contributions of this letter are: (i) formulating a reach-avoid game under incomplete information where Attackers' heterogeneous maximum speeds are publicly disclosed to fall within continuous ranges; (ii) proposing a critical speed pair framework to characterize when different heterogeneous speed configurations induce opposite optimal capture orders and to analyze the Defender's guessing behavior; and (iii) characterizing dilemma conditions under which Attackers can conceal their true capability configuration through an information-limiting strategy for improved payoff.

The letter is organized as follows. The problem is formulated in Section \ref{section:problem_formulation}, following preliminaries in Section \ref{section:chapter_complete_information}.  The concepts of critical speed pairs are introduced in Section \ref{section:critical_speed}. The main results are presented in Section \ref{section:main_result} with the visualization of dilemma conditions in Section \ref{section:visualization}.

\textit{Notations}. 
The unit vector with direction of nonzero vector $\mathbf{a}\in\mathbb{R}^n$ is denoted $\langle \mathbf{a}\rangle=\frac{\mathbf{a}}{\|\mathbf{a}\|}$.
For real numbers $a, b, c, d$, notation $(a,b)\succ (c,d)$ means $a\ge c,b\ge d$ with at least one inequality strict. Similarly, $(a,b)\neq(c,d)$ means at least one of $a\neq c$ or $b\neq d$ holds.

\section{Problem Formulation}\label{section:problem_formulation}
Consider a reach-avoid game in Euclidean plane $\mathbb{R}^2$, involving two Attackers ($A_1, A_2$) and one Defender ($D$), as illustrated in Fig. \ref{fig:scenario}.
The Attackers seek to reach the target region $\mathcal{T}=\{[x,y]\in\mathbb{R}^2|y\le0\}$ without capture, while the Defender aims to intercept them before they reach the target.
\begin{figure}[htbp]
        \centering
        \includegraphics[width=0.3\textwidth]{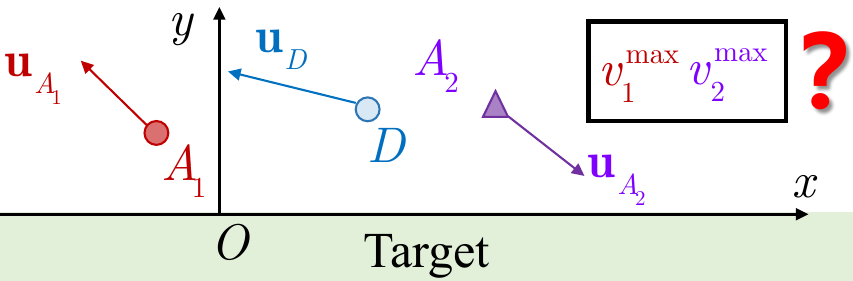}
        \caption{The 2A1D reach-avoid game scenario.
        }
        \label{fig:scenario}
\end{figure}

Let $\mathbf{x}_{A_i}=[x_{A_i},y_{A_i}],\mathbf{x}_D=[x_D,y_D]$ denote the positions of Attackers $A_i,i=1,2$, and the Defender $D$, respectively.
Following simple motion models \cite{isaacs1999differential}, the kinematic equations are
\begin{equation}
        \dot{\mathbf{x}}_{A_i}= \mathbf{u}_{A_i},i=1,2, \dot{\mathbf{x}}_D= \mathbf{u}_D. \label{dynamics} \\
\end{equation}
where $\mathbf{u}_{A_i}\in\{\mathbf{u}\in\mathbb{R}^2|\|\mathbf{u}\|\le v_i^{\text{max}}\}$, and $\mathbf{u}_D\in\{\mathbf{u}\in\mathbb{R}^2|\|\mathbf{u}\|\le1\}$ are the control inputs. Here, $v_i^{\text{max}}>0,i=1,2$ represent Attacker $A_i$'s maximum speed, while the Defender's maximum speed is normalized to unity. The game state is $\mathbf{x}=[\mathbf{x}_{A_1},\mathbf{x}_{A_2},\mathbf{x}_{D}]$.

Attacker $A_i$ is captured by the Defender before reaching the target boundary if $ \exists t:\mathbf{x}_{A_i}(t)=\mathbf{x}_{D}(t)$ and $\mathbf{x}_{A_i}(\tau)\notin\mathcal{T}, \forall \tau\in[0,t]$. Conversely, $A_i$ successfully reaches the target if $ \exists t:\mathbf{x}_{A_i}(t)\in\mathcal{T}$, and $\mathbf{x}_{A_i}(\tau)\neq\mathbf{x}_{D}(\tau),\forall \tau\in[0,t]$. Once an Attacker is captured or reaches the target at time $t_{A_i}$, it is removed from the game. The game continues until it terminates at $t_T=\max\{t_{A_1},t_{A_2}\}$.

The payoff function, representing the number of Attackers captured, is defined as follows.
\begin{equation}
        J(\mathbf{x}_0;\gamma_A,\gamma_D) = \left\{\begin{array}{cl}
        0 & \text{if no Attacker is captured,} \\
        1 & \text{if one Attacker is captured,} \\
        2 & \text{if two Attackers are captured.}
        \end{array}\label{payoff_function}\right.
\end{equation}
where $\mathbf{x}_0=[\mathbf{x}^0_{A_1},\mathbf{x}^0_{A_2},\mathbf{x}^0_{D}]$ is the initial state, 
$\gamma_A=\{\mathbf{u}_{A_1},\mathbf{u}_{A_2}\},\gamma_D=\{\mathbf{u}_D\}$ are the Attackers' and Defender's strategies, respectively.
The Attackers aim to minimize while the Defender seeks to maximize this payoff. Given the game state $\mathbf{x}$ and the Attackers' maximum speeds, the value of the game is denoted by 
\begin{equation}
        J^*(\mathbf{x},v_1^{\text{max}},v_2^{\text{max}})\triangleq\min_{\gamma_A}\max_{\gamma_D}J(\mathbf{x},\gamma_A,\gamma_D).\label{optimal_payoff}
\end{equation}

The state feedback information structure, under which all players make their decisions only based on the instantaneous game state $\mathbf{x}(t)$, is adopted. We make the following assumptions regarding Attackers' speed information.

\begin{assumption}\label{assumption:maximum speed}
The maximum speeds $v_i^{\text{max}}$ of the Attackers are private to the Attackers. However, it is common knowledge that $v_i^{\text{max}}\geq v_i^{\text{pub}}$, where $v_i^{\text{pub}}\in(0,1)$ are publicly known lower bounds.
\end{assumption}
\begin{assumption}\label{assumption:observe}
        The Defender can observe the accurate historical speed of the Attackers.
\end{assumption}

In this letter, we seek to address the following question: under what conditions can Attackers enhance their payoffs by creating dilemmas via deceptive strategy, while Defender is forced to make risky guesses?

\section{Preliminaries} \label{section:chapter_complete_information}
This section introduces the barriers of the reach-avoid game under complete information, where the Attackers' maximum speed $v_i^{\text{max}}<1$ is common knowledge among all players. 

\subsection{Barrier of one-Attacker one-Defender (1A1D) game}

The barrier of the 1A1D game is the surface separating the winning regions of the Attacker and Defender \cite{isaacs1999differential,yan2023multiplayer}. Here, the winning region of each player is a set of states from which that player can guarantee to win regardless of the opponent's strategy. 

Let the state in 1A1D game be $\mathbf{z}=[\mathbf{x}_A,\mathbf{x}_D]$, where $\mathbf{x}_A$ and $\mathbf{x}_D$ represent the positions of the Attacker and the Defender, respectively. Define the function $Y(\mathbf{z},v)\triangleq\frac{y_A-v^2y_D-v\|\mathbf{x}_A-\mathbf{x}_D\|}{1-v^2}$. 

\begin{lemma}[\cite{garcia2019strategies}]\label{lemma:barrier_one}
Given $v^{\text{max}}\in(0,1)$, the 1A1D game barrier is characterized by
        \begin{equation}
                \mathcal{B}_1\triangleq\{\mathbf{z}|Y(\mathbf{z},v^{\text{max}})=0\},\label{full_barrier_one_Attacker}\\
        \end{equation}
with the Attacker's winning region as $\mathcal{R}^0_{1}\triangleq\{\mathbf{z}|Y(\mathbf{z},v^{\text{max}})<0\}$, and the Defender's winning region as $\mathcal{R}^1_{1}\triangleq\{Y(\mathbf{z},v^{\text{max}})>0\}$.
Moreover, when players are positioned on the barrier $\mathcal{B}_1$, their optimal strategies are
\begin{align}
        \mathbf{u}^*_D&=\boldsymbol\eta_D(\mathbf{z},v^{\text{max}}), \label{optimal_strategy_one_Defender_barrier}\\
        \mathbf{u}^*_A&=v^{\text{max}}\boldsymbol\eta_A(\mathbf{z},v^{\text{max}}).\label{optimal_strategy_one_Attacker_barrier}
\end{align}
where $\boldsymbol\eta_M(\mathbf{z},v)\triangleq\langle \frac{\mathbf{x}_A-v^2\mathbf{x}_D}{1-v^2}-\frac{v\|\mathbf{x}_A-\mathbf{x}_D\|}{1-v^2}\cdot[0,1]-\mathbf{x}_M\rangle$ for $M=A,D$.
\end{lemma}

\subsection{Barriers of 2A1D game}\label{section:complete_information_two_attackers}
Given that the Defender captures the Attacker $A_i,i=1,2$ first, we consider the barrier in 2A1D game that separates two distinct state spaces: one from which the Defender can guarantee capturing both Attackers, and another from which at least one Attacker can guarantee reaching the target.
Let $\mathbf{z}_i\triangleq[\mathbf{x}_{A_i},\mathbf{x}_D]$ represent the state vector containing positions of Attacker $A_i$ and the Defender $D$ for $i=1,2$.
Given the capture order, the game state $\mathbf{x}$ and the Attackers' maximum speeds $v_i, v_j$, consider the lowest $y$-coordinate that the Attacker $A_j,j=\{1,2\}\setminus \{i\}$ can potentially reach \cite{deng2023reach}
\begin{equation}
        \hat J_i(\mathbf{x},v_i,v_j)\triangleq\min_{\theta\in(-\pi,\pi]\atop \phi\in(-\pi,0)}Y([\hat{\mathbf{x}}_i(\mathbf{x},v_i,v_j,\phi), \tilde{\mathbf{x}}(\mathbf{z}_i,v_i,\theta)],v_j)\label{y_minimum_2}
\end{equation}
where $\tilde{\mathbf{x}}(\mathbf{z}_i,v_i,\theta)\triangleq \frac{\mathbf{x}_{A_i}-v_i^2\mathbf{x}_D}{1-v_i^2}-\frac{v_i\|\mathbf{x}_{A_i}-\mathbf{x}_D\|}{1-v_i^2}\cdot[\cos\theta,\sin\theta]$, $\hat{\mathbf{x}}_i(\mathbf{x},v_i, v_j,\phi)\triangleq \mathbf{x}_{A_j}+v_j\|\tilde{\mathbf{x}}(\mathbf{z}_i,v_i,\theta)-\mathbf{x}_D\|\cdot[\cos\phi,\sin\phi]$.

There may exist multiple angle pairs that minimize \eqref{y_minimum_2} \cite{deng2023reach}. We assume players can compute the minimum value and corresponding angle pairs, denoted by $\theta^i_l(\mathbf{x},v_i,v_j),\phi^i_l(\mathbf{x},v_i,v_j),l=1,\cdots,m_i,m_i\ge1$, using established computational approaches\cite{garcia2019strategies,liu2013evasion}. 
\begin{lemma}[\cite{yan2021cooperative,deng2023reach}]
        Given that the Defender captures the Attacker $A_i$ first, the barrier for 2A1D game given the capture order is $\mathcal{B}_{2,i}\triangleq\mathcal{\bar B}_{2,i}\cup\mathcal{\hat B}_{2,i}$, where 
        \begin{align}
                \mathcal{\bar B}_{2,i}\triangleq\{&\mathbf{x}|\mathbf{z}_i\in\mathcal{R}_1^1, \hat J_i(\mathbf{x},v_i^{\text{max}},v_j^{\text{max}})=0\},\label{barrier_two_1_2}\\
                \mathcal{\hat B}_{2,i}\triangleq\{&\mathbf{x}|\mathbf{z}_i\in\mathcal{B}_1, \hat J_i(\mathbf{x},v_i^{\text{max}},v_j^{\text{max}})>0\},i=1,2.\label{barrier_two_1_2_2}
        \end{align}
        Additionally, if $\mathbf{x}\in\mathcal{\bar B}_{2,i}$, the optimal strategies of all players are 
        \begin{align}
                \mathbf{u}_{A_i}&=v_i^{\text{max}}\boldsymbol\zeta^{i,l}_{A_i}(\mathbf{x},v_i^{\text{max}},v_j^{\text{max}})\label{optimal_a1_first_phase} \\ 
                \mathbf{u}_{A_j}&=v_j^{\text{max}}\boldsymbol\xi_{A_j}^{i,l}(\mathbf{x},v_i^{\text{max}},v_j^{\text{max}}) \label{optimal_a2_first_phase}\\
                \mathbf{u}_D&=\boldsymbol\zeta^{i,l}_D(\mathbf{x},v_i^{\text{max}},v_j^{\text{max}}),l=1,\cdots,m_i\label{optimal_d_first_phase}
        \end{align}
        where $\boldsymbol\zeta^{i,l}_M(\mathbf{x},v_i,v_j)\triangleq\langle\tilde{\mathbf{x}}(\mathbf{z}_i,v_i,\theta^i_l(\mathbf{x},v_i))-\mathbf{x}_M\rangle$ for $M=A_i,D$, and $\boldsymbol\xi_{A_j}^{i,l}(\mathbf{x},v_i,v_j)\triangleq\langle\hat{\mathbf{x}}_i(\mathbf{x},v_i,v_j,\phi^i_l(\mathbf{x},v_i,v_j))-\mathbf{x}_{A_j} \rangle$.
\end{lemma}

\begin{remark}
        When multiple angle pairs minimize \eqref{y_minimum_2},  Attackers may select any optimal pair yielding identical $\hat J_i$ value, while the Defender faces a momentary dilemma \cite{isaacs1999differential}.
\end{remark}

\section{The critical speeds and critical speed pairs}\label{section:critical_speed}

Following \cite{wu2025deception,wu2026inducing}, critical speeds, at which exactly one or two Attackers are captured upon reaching the target, are essential for characterizing strategic interactions between Attackers and Defender. 

\begin{definition}\label{definition:critical_mu}
        (\textit{Critical speed for 1A1D game})\cite{wu2025deception,wu2026inducing}. The speed $\mu_i,i=1,2$ is the critical speed of the 1A1D game if the Attacker $A_i$ with maximum speed $v_i^{\text{max}}>\mu_i$ cannot be captured regardless of the Defender's strategy,
        and the Attacker $A_i$ with $v_i^{\text{max}}<\mu_i$ can be captured regardless of the strategy of $A_i$.
\end{definition}

However, critical speed for 2A1D game proposed in \cite{wu2025deception,wu2026inducing} fails to address Attackers' heterogeneous maximum speeds. To overcome this limitation, we introduce the concept of critical speed pairs as follows.

\begin{definition}\label{definition:critical_lambda}
        (\textit{Critical speed pair for 2A1D game with given capture order}). Given that the Defender captures $A_i$ first, the speed pair $(\lambda_i,\lambda_j)$ is called the critical speed pair of the 2A1D game if at most one Attacker can be captured regardless of the Defender's strategy for $(v_i^{\text{max}},v_j^{\text{max}})\succ (\lambda_i,\lambda_j)$ and both Attackers can be captured regardless of the strategy of the Attackers for $(v_i^{\text{max}},v_j^{\text{max}})\prec (\lambda_i,\lambda_j)$.    
\end{definition}

\subsection{Computation of the critical speed (pairs)}

Critical speed (pairs) can be computed by \eqref{full_barrier_one_Attacker} and \eqref{barrier_two_1_2_2}.

\begin{lemma}[\cite{wu2026inducing}]\label{lemma_critical}
        The function $Y(\mathbf{z}_i,v)$ is continuous and strictly decreasing on $v\in[0,1)$. Moreover, the critical speed $\mu_i$ for 1A1D game is the unique solution in $v\in (0,1)$ to the equation
        \begin{equation}
                Y(\mathbf{z}_i,v)=0.\label{compute_mu}
        \end{equation}
\end{lemma}

\begin{lemma}\label{lemma_L_decrease}
        The function $\hat J_i(\mathbf{x},v_i,v_j)$ is continuous and strictly decreasing in $v_i\in[0,\mu_i]$ for fixed $v_j\in[0,\mu_j]$, and $v_j\in[0,\mu_j]$ for fixed $v_i\in[0,\mu_i]$, where $\mu_i,\mu_j$ are the solutions of \eqref{compute_mu}. Then, for any given $\lambda_i\in[0,\mu_i],i=1,2$,  $\lambda_j$ is the unique solution of $v_j\in(0,\mu_j),j=\{1,2\}\setminus i$ to the equation
        \begin{equation}
                \hat J_i(\mathbf{x},\lambda_i,v_j)=0.\label{compute_pair}
        \end{equation}        
        Moreover, the pair $(\lambda_i,\lambda_j)$ with $\lambda_i<\mu_i$ is a critical speed pair for 2A1D game given capture order that the Defender captures $A_i$ first.
\end{lemma}

\begin{proof}
(1) Properties of function $\hat J_i(\mathbf{x},v_i,v_j)$.
Consider $\tilde\theta\in[-\pi,\pi],\tilde\phi\in[-\pi,0]$, then $
\hat J_i(\mathbf{x},v_i,v_j)=\min_{\tilde\theta,\tilde\phi}Y([\hat{\mathbf{x}}_i(\mathbf{x},v_i,v_j,\tilde\phi), \tilde{\mathbf{x}}(\mathbf{z}_i,v_i,\tilde\theta)],v_j)$
holds based on the definition of $\hat J_i$. The continuity of $\hat J_i(\mathbf{x},v_i,v_j)$ with respect to $\mathbf{x}, v_i,v_j$ follows from Berge Maximum Theorem \cite{Beavis_Dobbs_1990}.

To establish the strict monotonicity, we introduce an auxiliary 2A1D game where Defender $D$ with maximum speed $1$ captures $A_i$ with maximum speed $v_i$ first and then $A_j$ with maximum speed $v_j$, without a target region. The payoff function is $y_{A_j}(t_{A_j})$ (the $y$-coordinate where $A_j$ is captured), which Attackers seek to minimize and Defender aims to maximize.

Therefore, we have the following result based on \cite{yan2021cooperative}.
Given state $\mathbf{x}$, the maximum speeds of the Attackers $v_i,v_j$ and assuming the Defender captures $A_i$ first, under the payoff function $y_{A_j}(t_{A_j})$, all players' optimal strategies are constant headings at maximum speed for $t\le t_{A_i}$.

Let $L_i(\mathbf{x},v_i,v_j)\triangleq \min_{\gamma_A}\max_{\gamma_D} y_{A_j}(t_{A_j})$ denote the value of this game.
To compute $L_i(\mathbf{x},v_i,v_j)$, we follow the optimal strategies and let the Attacker $A_i$ move at maximum speed toward point $\tilde{\mathbf{x}}(\mathbf{z}_i,v_i,\theta)$, which is a point on the boundary of the $A_1$'s dominant regions \cite{yan2021cooperative} to delay the capture of $A_j$ while the Attacker $A_j$ moves in direction $\phi\in(-\pi,\pi]$ for $t\le t_{A_i}$. Thus, the Defender must move toward $\tilde{\mathbf{x}}(\mathbf{z}_i,v_i,\theta)$ to capture $A_i$.
The time to capture $A_i$ is $t_{A_i}=\|\tilde{\mathbf{x}}(\mathbf{z}_i,v_i,\theta)-\mathbf{x}_D\|$,
and the Attacker $A_j$'s position at $t_{A_i}$ is $\hat{\mathbf{x}}_i(\mathbf{x},v_i, v_j,\phi)$.
Consequently, the smallest $y$-coordinate that the Attacker $A_j$ can potentially reach can be calculated as $L_i(\mathbf{x},v_i,v_j)=\hat J_i(\mathbf{x},v_i,v_j)$.
        
Then, we can back to prove the property of function $\hat J_i(\mathbf{x},v_i,v_j)$. Given fixed $v_j$, suppose $v_i^1,v_i^2\in[0,\mu_i]$ with $v_i^1<v_i^2$, then $\hat J_i(\mathbf{x},v_i^1,v_j)>\hat J_i(\mathbf{x},v_i^2,v_j)$ must be hold.
Otherwise, $\hat J_i(\mathbf{x},v_i^1,v_j)\le\hat J_i(\mathbf{x},v_i^2,v_j)$ indicates $L_i(\mathbf{x},v_i^1,v_j)\le L_i(\mathbf{x},v_i^2,v_j)$, which follows that the Attackers move with the relatively slow speed $v_i^1$ can reach the point with smaller $y$-coordinate, which leads to the contradiction with optimal strategies. The strict monotonicity with respect to $v_j$ given fixed $v_i$ follows similarly. 

(2) Property of the critical speed pairs.
From Lemma \ref{lemma_critical}, the critical speed $\mu_i,i=1,2$ exists. For every given speed $\lambda_i\in[0,\mu_i],i=1,2$, there are $\hat J_i(\mathbf{x},\lambda_i,0)=y_{A_j}>0$ and $\hat J_i(\mathbf{x},\lambda_i,\mu_j)<Y(\mathbf{z}_j,\mu_j)=0$. 
Therefore, there exists a unique $\lambda_j\in(0,\mu_j)$ such that $\hat J_i(\mathbf{x},\lambda_i,\lambda_j)=0$. For $(v_i,v_j)\prec (\lambda_i,\lambda_j)$ with $\lambda_i<\mu_i$, it has $\hat J_i(\mathbf{x},v_i,v_j)>0$ and $Y(\mathbf{z}_i,v_i)>0$, indicating that the Defender can successfully capture both the Attackers. For $(v_i,v_j)\succ (\lambda_i,\lambda_j)$, it has $\hat J_i(\mathbf{x},v_i,v_j)<0$, which implies that the Defender cannot capture the Attacker $A_j$. Therefore, $(\lambda_i,\lambda_j)$ is the critical speed pair for 2A1D game given the Defender captures $A_i$ first.
\end{proof}

From Lemma \ref{lemma_L_decrease}, there exists a unique continuous, strictly decreasing function $h_i:\mathbb{R}^6\times\mathbb{R}\to\mathbb{R}$ such that $\hat J_i(\mathbf{x},v_i,h_i(\mathbf{x},v_i))=0$ for all $v_i\in[0,\mu_i]$. 
The inverse function  $h^{-1}_i(\mathbf{x},v_j)$ for $v_j\in[h_i(\mathbf{x},\mu_i),h_i(\mathbf{x},0)]$ with fixed $\mathbf{x}$ exists. 
For clarity, we omit the arguments of $\mathbf{x}$ for functions $h_i,h_i^{-1}$ when no ambiguity arises.
Define the function $h(v_1)\triangleq\max\{\bar h_2(v_1),h_1(v_1)\}$, where
\begin{equation}
        \bar h_2(v_1)\triangleq\left\{\begin{array}{ll}
                \mu_2, & v_1\in(0,h_2(\mu_2)),\\
                h_2^{-1}(v_1), &v_1\in[h_2(\mu_2),h_2(0)],
        \end{array}\right.
\end{equation}
as the boundary in $(v_1,v_2)$-plane separating the regions where
$J^*(\mathbf{x},v_1,v_2)=2$ and $J^*(\mathbf{x},v_1,v_2)\le1$, as seen in Fig. \ref{fig:diagram_v1v2_plane}.
\begin{lemma}\label{lemma:h_property}
        The function $h(v_1)$ is continuous and strictly decreasing on $v_1\in[h_2(\mu_2),\mu_1]$. Moreover, for all $v_1\in(h_2(\mu_2),\mu_1)$, $J^*(\mathbf{x},v_1^{\text{max}},v_2^{\text{max}})=2$ when $(v_1^{\text{max}},v_2^{\text{max}})\prec (v_1,h(v_1))$, and $J^*(\mathbf{x},v_1^{\text{max}},v_2^{\text{max}})\le1$ when $(v_1^{\text{max}},v_2^{\text{max}})\succ (v_1,h(v_1))$.
\end{lemma}
\begin{proof}
Since $\bar h_2(h_2(0))=h_2^{-1}(h_2(0))=0$, the function $\bar h_2(v_1)$ is continuous and strictly decreasing on $v_1\in[h_2(\mu_2),\mu_1]$.
Combined with $h(v_1)\triangleq\max\{\bar h_2(v_1),h_1(v_1)\}$ and the established properties of $h_1(v_1)$, the function $h(v_1)$ is continuous and strictly decreasing on $v_1\in[h_2(\mu_2),\mu_1]$.

Given the game state $\mathbf{x}$ and the speed $v_1\in(h_2(\mu_2),\mu_1)$, assume $h_1(v_1)\ge \bar h_2(v_1)$ holds without loss of generality. 

If $(v_1^{\text{max}},v_2^{\text{max}})\prec (v_1,h(v_1))$ holds, then $(v_1^{\text{max}},v_2^{\text{max}})\prec (v_1,h_1(v_1))$. The definition of $h_1$ yields $\hat J_1(\mathbf{x},v_1,h_1(v_1))=0$. Consequently, $\hat J_1(\mathbf{x},v_1^{\text{max}},v_2^{\text{max}})>0$ by Lemma 5. Thus, the Defender can capture $A_1$ by adopting (12) first and then capture $A_2$ by adopting (5) regardless of the Attackers' strategy. Therefore, $J^*(\mathbf{x},v_1^{\text{max}},v_2^{\text{max}})=2$ holds for $(v_1^{\text{max}},v_2^{\text{max}})\prec (v_1,h(v_1))$. 

For $(v_1^{\text{max}},v_2^{\text{max}})\succ (v_1,h(v_1))=(v_1,h_1(v_1))$, it has that $\hat J_1(\mathbf{x},v_1^{\text{max}},v_2^{\text{max}})<0$. Additionally, $\hat J_2(\mathbf{x},v_2^{\text{max}},v_1^{\text{max}})<\hat J_2(\mathbf{x},h_1(v_1),v_1)\le \hat J_2(\mathbf{x},\bar h_2(v_1),v_1)=0$
also holds true. Let $i^*\triangleq \arg\max_{i=1,2}\hat J_i(\mathbf{x}, v_i^{\text{max}},v_j^{\text{max}})$ 
denote the index with the larger $\hat J_i$ value. Thus, if the Attackers adopt (10) and (11) with $i=i^*$, then the value of $\max_{i=1,2} \hat J_i(\mathbf{x}, v_i^{\text{max}},v_j^{\text{max}})$ remains non-increasing until one Attacker is removed based on the equilibrium property of the auxiliary game considered in the proof of Lemma \ref{lemma_L_decrease}. This ensures the Defender can capture at most one Attacker, regardless of the Attackers' strategy. Therefore, $J^*(\mathbf{x},v_1^{\text{max}},v_2^{\text{max}})\le1$.        
\end{proof}

\begin{figure}[htbp]      
        \centering
        \includegraphics[width=0.98\linewidth]{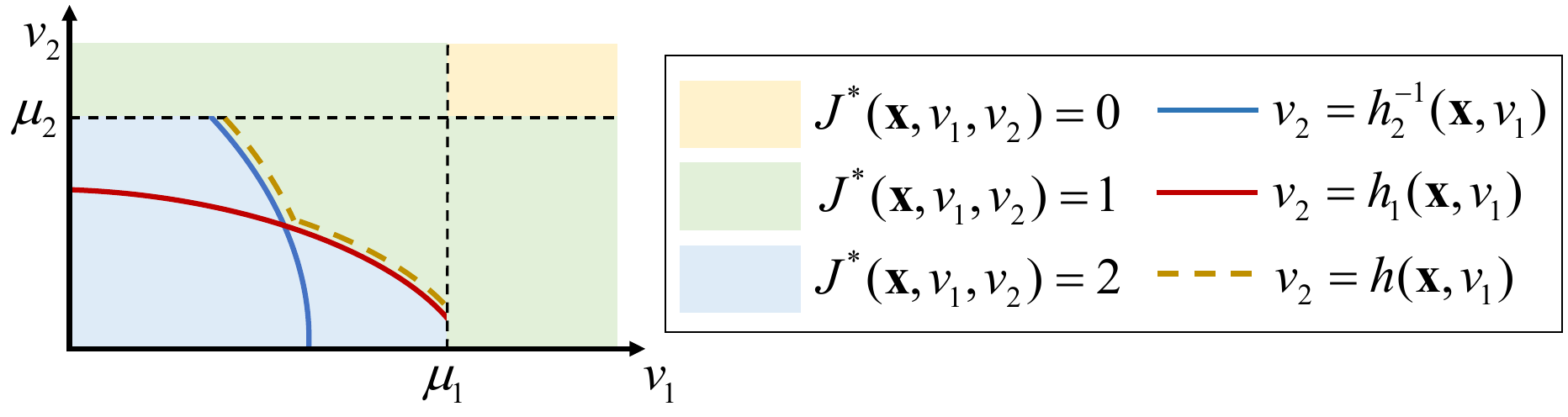}
        \caption{The illustration of the functions $h_1,h_2$ and $h$.
        }
        \label{fig:diagram_v1v2_plane}
\end{figure}

While explicit expressions for $h_i,h^{-1}_i,i=1,2$, cannot be obtained, their values can be efficiently computed using the bisection method. Note that critical speeds and pairs depend on $\mathbf{z}_i$, and $\mathbf{x}$, respectively, the arguments will be omitted when no ambiguity arises.

\subsection{The rate of the critical speed (pairs)}

Since critical speed (pairs) evolve with game state, we derive their time derivatives for subsequent analysis.
\begin{lemma}[\cite{wu2026inducing}]\label{lemma:rate_mu_k}
        The rate of the critical speed $\mu_i$ for $i=1,2$ is
        \begin{equation}
                \dot{\mu}_i = C_{\mu,i1}[\mathbf{u}_{A_i}\mathbf{C}_{\mu,i2}+\mathbf{u}_D\mathbf{C}_{\mu,i3}],\label{rate_mu_k_1}
        \end{equation}
        where $C^{-1}_{\mu,i1}=-\frac{2\mu_i}{(1-\mu_i^2)^2}(y_{A_i}-y_D)+\frac{1+\mu_i^2}{(1-\mu_i^2)^2}\|\mathbf{x}_{A_i}-\mathbf{x}_D\|>0$, 
        $\mathbf{C}^T_{\mu,i2}=\frac{1}{1-\mu_i^2}([0,1]+\mu_i\langle \mathbf{x}_D-\mathbf{x}_{A_i}\rangle)$,
        $\mathbf{C}^T_{\mu,i3}=-\frac{\mu_i}{1-\mu_i^2}([0,\mu_i]+\langle \mathbf{x}_D-\mathbf{x}_{A_i}\rangle)$.
\end{lemma}

\begin{lemma}\label{lemma:rate_lambda_i}
        Given the capture order that the Defender captures $A_i$ first and the critical speed pair $(\lambda_i,\lambda_j)$ for fixed $\lambda_i\in(h_j(\mu_j),h_j(0))$, if the angle pair that minimizes \eqref{y_minimum_2} is unique, i.e. $m_i=1$, then the rate of $\lambda_j$ is 
        \begin{equation}
                \dot{\lambda}_j=C_{\lambda,i1}\left[\mathbf{u}_{A_i}\mathbf{C}_{\lambda,i2}+\mathbf{u}_{A_j}\mathbf{C}_{\lambda,i3}+\mathbf{u}_D\mathbf{C}_{\lambda,i4}\right].\label{rate_lambda_i_1}
        \end{equation}
        Here, $C^{-1}_{\lambda,i1}=-\frac{\partial \hat J_i}{\partial v_j}|_{\lambda_i,\lambda_j}>0, \mathbf{C}^T_{\lambda,i2}=\frac{\partial \hat J_i}{\partial \mathbf{x}_{A_i}}|_{\lambda_i,\lambda_j},\mathbf{C}^T_{\lambda,i3}=\frac{\partial \hat J_i}{\partial \mathbf{x}_{A_j}}|_{\lambda_i,\lambda_j},\mathbf{C}^T_{\lambda,i4}=\frac{\partial \hat J_i}{\partial \mathbf{x}_D}|_{\lambda_i,\lambda_j}$. 
\end{lemma}
\begin{proof}
        If the minimum angles are unique,
then the function $\hat J_i(\mathbf{x},v_i,v_j)$ is differentiable with respect to $\mathbf{x}$ with  
$\langle\frac{\partial \hat J_i}{\partial\mathbf{x}_{A_i}}\rangle=-\boldsymbol\zeta^{i,1}_{A_i}(\mathbf{x},v_i,v_j),
\langle\frac{\partial \hat J_i}{\partial\mathbf{x}_{A_j}}\rangle=-\boldsymbol\xi^{i,1}_{A_j}(\mathbf{x},v_i,v_j),
\langle\frac{\partial \hat J_i}{\partial\mathbf{x}_D}\rangle=\boldsymbol\zeta^{i,1}_{D}(\mathbf{x},v_i,v_j)$ \cite{deng2023reach}.
Moreover, the function $\hat J_i(\mathbf{x},v_i,v_j)$ is differentiable with respect to $v_i,v_j$ since $\frac{\partial\hat J_i}{\partial \hat{\mathbf{x}}_1},\frac{\partial\hat J_i}{\partial \tilde{\mathbf{x}}}, \frac{\partial\hat{\mathbf{x}}_i}{\partial v_i}, \frac{\partial\tilde{\mathbf{x}}}{\partial v_i}$ and $\frac{\partial\tilde{\mathbf{x}}}{\partial v_j}$ exist. 
Consequently, $\frac{\partial \hat J_i}{\partial v_i}<0,\frac{\partial \hat J_i}{\partial v_j}<0$ holds due to Lemma 6.
Therefore, based on the fact that the equation $\hat J_i(\mathbf{x},\lambda_i,\lambda_j)=0$ holds for all $t$, differentiating with respect to time gives
$\dot{\hat J}_i(\mathbf{x},\lambda_i,\lambda_j)=\frac{\partial \hat J_i}{\partial \mathbf{x}_{A_i}}|_{\lambda_i,\lambda_j}\dot{\mathbf{x}}^T_{A_i}+ \frac{\partial \hat J_i}{\partial \mathbf{x}_{A_j}}|_{\lambda_i,\lambda_j}\dot{\mathbf{x}}^T_{A_j} + \frac{\partial \hat J_i}{\partial \mathbf{x}_D}|_{\lambda_i,\lambda_j}\dot{\mathbf{x}}^T_D+\frac{\partial \hat J_i}{\partial v_j}|_{\lambda_i,\lambda_j}\dot{\lambda}_j=0$.
Thus, Eq. (17) can be obtained by substituting the dynamics into this equation.
\end{proof}

\section{Main results}\label{section:main_result}

Based on critical speeds (pairs), this section characterizes the strategic behaviors in the 2A1D reach-avoid game under Assumptions \ref{assumption:maximum speed} and \ref{assumption:observe}.

\subsection{Mismatching order condition of the Defender}

In this subsection, we establish a necessary and sufficient condition under which information about $v_i^{\text{max}}$ would alter the Defender's optimal capture order. 

Given game state $\mathbf{x}$ and the Attackers' maximum speeds $v_i,v_j$, define the optimal payoff when the Defender captures Attacker $A_i$ first
\begin{equation}
        J_D^i(\mathbf{x},v_i,v_j)\triangleq\min_{\gamma_A}\max_{\gamma_D\in\Gamma_D^i(\mathbf{x},v_i)}J(\mathbf{x},\gamma_A,\gamma_D)\label{defender_payoff},
\end{equation}
where $\Gamma_D^i(\mathbf{x},v_i)\triangleq\{\{\langle\mathbf{\tilde{x}}(\mathbf{z}_i,v_i,\theta)-\mathbf{x}_D\rangle\},\theta\in(-\pi,\pi]\}$ represents strategies prioritizing pursuit of $A_i$ at maximum speed. The mismatching order state is defined as follows.

\begin{definition}\label{definition_mismatching}
The state $\mathbf{x}$ is a mismatching order state if there exist speeds $v_{1,1},v_{1,2}>v_1^{\text{pub}},v_{2,1},v_{2,2}>v_2^{\text{pub}}$ such that 
\begin{align}
  &[J_D^1(\mathbf{x},v_{1,1},v_{2,1})-J_D^2(\mathbf{x},v_{2,1},v_{1,1})]\times\nonumber\\ &[J_D^1(\mathbf{x},v_{1,2},v_{2,2})-J_D^2(\mathbf{x},v_{2,2},v_{1,2})]<0.   
\end{align}
The set of all mismatching order states is denoted $\mathcal{C}$.
\end{definition} 

\begin{remark}
        The mismatching order states identify when the Defender's optimal capture order depends on the Attackers' true maximum speeds. In such states, Attackers may conceal their capabilities to exploit information asymmetry. 
\end{remark}

The following theorem presents the necessary and sufficient condition for the mismatching order state.

\begin{theorem}\label{theorem:mismatching}
        The state $\mathbf{x}$ is a mismatching order state if and only if 
        \begin{equation}
             v_1^{\text{pub}}<\mu_1,v_2^{\text{pub}}<\mu_2 
        \end{equation}
\end{theorem}
\begin{proof}
        ($\Leftarrow$). 
        For speeds $v_{1,1}\in(v_1^{\text{pub}},\mu_1),v_{2,1}\in(\mu_2,1)$, we have $J_D^1(\mathbf{x},v_{1,1},v_{2,1})=1$. In the following, we aim to prove that there exist some $v_{1,1}\in(v_1^{\text{pub}},\mu_1),v_{2,1}\in(\mu_2,1)$ such that $J_D^2(\mathbf{x},v_{2,1},v_{1,1})=0$ holds.
        
        Consider the following strategy for Attackers: $A_2$ first moves towards the point $\mathbf{y}_1\triangleq(\frac{x_{A_2}-v_{2,1}^2x_D}{1-v_{2,1}^2},0)$, then remains at this position until reaching the target when the Defender is sufficiently close. For the Attacker $A_1$, it moves towards $\mathbf{y}_2\triangleq(x_{A_1}-\frac{v_{1,1}y_{A_1}}{\sqrt{1-v_{1,1}^2}},0)$ if $x_{A_1}\le\frac{x_{A_2}-\mu_2^2x_D}{1-\mu_2^2}$, otherwise it moves towards $\mathbf{y}_3\triangleq(x_{A_1}+\frac{v_{1,1}y_{A_1}}{\sqrt{1-v_{1,1}^2}},0)$. Under this strategy, the Defender adopting any $\gamma_D\in\Gamma_D^2(\mathbf{x},v_{2,1})$ is forced to approach $\mathbf{y}_1$ but cannot intercept $A_2$. 
We examine whether $A_1$ can be captured. 

For the case where $x_{A_1}\le\frac{x_{A_2}-\mu_2^2x_D}{1-\mu_2^2}$, the time for $A_1$ to reach $\mathbf{y}_2$ is $t_1(v_{1,1})\triangleq\frac{y_{A_1}}{v_{1,1}\sqrt{1-v_{1,1}^2}}$ and the minimum time for the Defender to reach $\mathbf{y}_2$ via $\mathbf{y}_1$ is 
\begin{align}
    &t_D(v_{1,1},v_{2,1})\triangleq\nonumber\\
    &\sqrt{(\frac{x_D-x_{A_2}}{1-v_{2,1}^2})^2+y_D^2}+|\frac{x_{A_2}-v_{2,1}^2x_D}{1-v_{2,1}^2}-x_{A_1}+\frac{v_{1,1}y_{A_1}}{\sqrt{1-v_{1,1}^2}}|. 
\end{align}
We define the function $f_1(v_1,v_2)\triangleq t_1(v_1)-t_D(v_1,v_2)$ to represent the time difference for $A_1$ with maximum speed and the Defender. Since $f_1(v_1,v_2)$ is a continuous function of $v_1,v_2$ and 
\begin{align}
    &f_1(\mu_1,\mu_2)\nonumber\\
    =&\sqrt{(x_{A_1}-x_D)^2+(1-\mu_1)^2y_D^2}+(x_{A_1}-x_{D})\nonumber\\
    &-\frac{x_{A_2}-x_D}{1-\mu_2^2}-\sqrt{(\frac{x_D-x_{A_1}}{1-\mu_2^2})^2+y_D^2}\nonumber\\
    \le&\sqrt{(\frac{x_{A_2}-x_D}{1-\mu_1^2})^2+(1-\mu_2^2)y_D^2}-\sqrt{(\frac{x_D-x_{A_2}}{1-\mu_2^2})^2+y_D^2}\nonumber\\
    <&0,
\end{align}
there exists a positive number $\delta_1>0$ such that $f_1(v_1,v_2)<0$ holds for $v_1\in(\mu_1-\delta_1,\mu_1),v_2\in(\mu_2,\mu_2+\delta_1)$. Hence, it yields that  $J_D^2(\mathbf{x},v_{1,1},v_{2,1})=0$ for $v_{1,1}\in(\mu_1-\delta_1,\mu_1),v_{2,1}\in(\mu_2,\mu_2+\delta_1)$. Through similar analysis for the case where $x_{A_1}>\frac{x_{A_2}-\mu_2^2x_D}{1-\mu_2^2}$, we can also establish that there exists a positive number $\delta_2>0$ such that $J_D^2(\mathbf{x},v_{2,1},v_{1,1})=0$ for $v_{1,1}\in(\mu_1-\delta_2,\mu_1),v_{2,1}\in(\mu_2,\mu_2+\delta_2)$. 
Consequently, $J_D^2(\mathbf{x},v_{2,1},v_{1,1})=0$ holds for some  $v_{1,1}\in(v_1^{\text{pub}},\mu_1),v_{2,1}\in(\mu_2,1)$.

        Similarly, for some $v_{1,2}\in(\mu_1,1),v_{2,2}\in(v_2^{\text{pub}},\mu_2)$, we obtain $J_D^1(\mathbf{x},v_{1,2},v_{2,2})-J_D^2(\mathbf{x},v_{2,2},v_{1,2})=-1$.
        
        ($\Rightarrow$). We proceed by contradiction and assume $v_1^{\text{pub}}\ge \mu_1$ without loss of generality. Then, $J_D^2(\mathbf{x},v_2,v_1)=1$ and $J_D^1(\mathbf{x},v_1,v_2)\le1$ for all $v_1\ge v_1^{\text{pub}}$ and $v_2\in(0,\mu_2)$, while $J_D^2(\mathbf{x},v_2,v_1)=J_D^1(\mathbf{x},v_1,v_2)=0$ for $v_2>\mu_2$. 
        Therefore, we have $J_D^1(\mathbf{x},v_1,v_2)\ge J_D^2(\mathbf{x},v_2,v_1)$ for all $ v_1\ge v_1^{\text{pub}}, v_2\ge v_2^{\text{pub}}$, contradicting the mismatching order condition. 
\end{proof}

\subsection{Deceptive strategy of the Attackers}\label{section:strategy}

In this subsection, we derive the deceptive strategy for the Attackers to exploit the information advantage in mismatching order states.  
Note that Attackers employ deception only when tangible benefits are available. When $J^*(\mathbf{x},v_1^{\text{max}},v_2^{\text{max}})=0$, no improvement is possible, therefore eliminating deceptive motivation. Therefore, given state $\mathbf{x}$, deception is beneficial if and only if $J^*(\mathbf{x},v_1^{\text{max}},v_2^{\text{max}})\ge1$.
We define the motivation state set $\mathcal{M}$ as states satisfying this condition. Consequently, the Attackers are motivated to deceive if and only if $v_1^{\text{max}}\le\mu_1$ or $v_2^{\text{max}}\le\mu_2$.

To deceive the Defender to adopt a suboptimal strategy, the idea of information-limiting strategy is formulated in \cite{shishika2024deception}, which is proposed to prevent the Defender from inferring the Attackers' speed based on their observed Attackers' historical trajectories and speed. In this letter, we consider an attack strategy  
\begin{align}
        &\mathbf{u}_{A_1}=v_1^{\text{pub}}\boldsymbol\eta_{A_1}(\mathbf{z}_1,\mu_1)\label{deceptive_strategy_a1}\\
        &\mathbf{u}_{A_2}=v_2^{\text{pub}}\boldsymbol\eta_{A_2}({\mathbf{z}_2},\mu_2),\label{deceptive_strategy_a2}
\end{align} 

Thus, this slow-speed strategy clearly qualifies as an information-limiting strategy since it is independent of $v_1^{\text{max}},v_2^{\text{max}}$.

However, the information-limiting slow-speed strategies may increase the payoff, creating strategic risk. 
Since $J^*(\mathbf{x}(t),v_{\text{max}})=2$ already represents the maximum payoff, this case only occurs when $J^*(\mathbf{x}(t),v_1^{\text{max}},v_2^{\text{max}})=1$. To address this type of risk, we introduce warning states for the Attackers, in which the maximum speed strategy should be employed to prevent any increase of the payoff.     

\begin{definition}
        The state $\mathbf{x}$ is a warning state if 
        \begin{equation}
                -\varepsilon\le \max\{\hat J_1(\mathbf{x},v_1^{\text{max}},v_2^{\text{max}}),\hat J_2(\mathbf{x},v_2^{\text{max}},v_1^{\text{max}})\}<0,\label{warning_condition}
        \end{equation}
where $\varepsilon>0$ is a predetermined warning threshold.
The set of all warning states is denoted $\mathcal{W}$.
\end{definition}

\begin{remark}
The condition \eqref{warning_condition} also indicates that 
once the game state transitions from $J^*=2$ to $J^*=1$, Attackers must reveal their true speed to prevent the Defender from revisiting its strategy, which may reestablish the payoff of 2.
\end{remark}

\subsection{The guess of the Defender}\label{section:guess}

As pointed out in \cite{wu2025deception}, the essence of the Defender's guessing behavior is to exclude possible maximum speeds based on the observed game trajectory. Motivated by this observation, we extend the notion of a guess to the case of heterogeneous maximum speeds.

\begin{definition}\label{definition:guess}
        Given a trajectory of game states $\mathbf{x}_{[0,T]}=\{\mathbf{x}(t)|t\in[0,T]\}$ with $T>0$, the Defender is said to have made a guess that $(v_1^{\text{max}},v_2^{\text{max}})\neq(\nu_1,\nu_2)$ at the instant $t<T$ if there exist $\nu_1\in [v_1^{\text{pub}},1),\nu_2\in [v_2^{\text{pub}},1)$ and $\delta>0$ such that $J^*(\mathbf{x}(t^\prime),\nu_1,\nu_2)>J^*(\mathbf{x}(t^{\prime\prime}),\nu_1,\nu_2)$,
        $\forall t^\prime\in(0,t),t^{\prime\prime}\in(t,t+\delta)$. 
\end{definition}

Definition \ref{definition:guess} characterizes scenarios where the Defender's control input $\mathbf{u}_D$ at time $t$ deteriorates optimal payoff if Attackers' maximum speed is $(\nu_1,\nu_2)$. This decrease indicates that the Defender has excluded speed $(\nu_1,\nu_2)$, i.e., made a guess that $(v_1^{\text{max}},v_2^{\text{max}})\neq(\nu_1,\nu_2)$. If this exclusion is incorrect, Attackers can immediately exploit their true maximum speed $(v_1^{\text{max}},v_2^{\text{max}})$ for better payoff without deception. Otherwise, Attackers should maintain their deceptive strategy. To render this definition more tractable, the following lemmas are provided. 
\begin{lemma}\label{lemma:guess}
        Given a trajectory of game states $\mathbf{x}_{[0,T]}=\{\mathbf{x}(t)|t\in[0,T]\}$ with $T>0$ such that $\mathbf{x}\in\mathcal{C},\forall t\in[0,T]$, if
        \begin{equation}
        \min_{t\in[0,T]}\mu_i(t)<\mu_i(0)\label{mu_i_decrease},   
        \end{equation}
        the Defender has made guesses that $(v_i^{\text{max}},v_j^{\text{max}})\notin\mathcal{S}_{\mu_i}(\mathbf{x}_{[0,T]})\triangleq\{(\nu_i,\nu_j)|\nu_i=\mu_i(t)\in(\mu_i(t^\prime),\mu_i(0)],\nu_j>\mu_j(t),\exists\delta_{\mu}>0, \forall t\in[0,T),t^\prime\in(t,t+\delta_{\mu})\}$. 
\end{lemma}
\begin{proof}
        Since $\mathbf{x}\in\mathcal{C},\forall t\in[0,T]$, it has that $v_i^{\text{known}}<\mu_i(t), \forall t\in[0, T]$ for $i=1,2$.
If the condition \eqref{mu_i_decrease} holds, then for every value $\hat\nu_i\in(\mu_{i,\min}^{[0,T]},\mu_i(0)]$, there exists a set of the time instants $\Lambda=\{t_{1},t_{2},\cdots,t_{p}\}$, where $p$ is a positive integer, such that $\mu_i(t_{\lambda})=\hat\nu_i$ and $\dot{\mu}_i(t_{\lambda})<0$, $\forall t_{\lambda}\in\Lambda$.
Then there exists a positive number $\bar\delta_k>0$ such that $\mu_i(t)<\hat\nu_i$ for $t\in(t_{k}, t_{k}+\bar\delta_k)$. Then $\forall\epsilon_k>0$, there exists a positive number $\delta_k\in(0,\bar\delta_k]$ such that $\mu_j(t_k)+\epsilon_k>\mu_j(t)$ for $t\in(t_{k}, t_{k}+\delta_k)$.

For $k=1$, it has that $\hat\nu_i=\min_{t\in[0,t_1]}\mu_i(t)$. Correspondingly, there is $\forall  t^{\prime}\in(0,t_1),t^{\prime\prime}\in(t_1,t_1+\delta_1)$, for all $\nu_j\ge\mu_j(t_1)+\epsilon_k, j =\{1,2\}\setminus i$,
\begin{equation}
    J^*(\mathbf{x}(t^{\prime}), \beta_1,\beta_2)\ge1>J^*(\mathbf{x}(t^{\prime\prime}), \beta_1,\beta_2)=0. \label{lemma9_proof1}
\end{equation} 
where
\begin{equation}
    (\beta_1,\beta_2)=\left\{ \begin{array}{cc}
        &(\hat\nu_1,\nu_2),i=1\\
        &(\nu_1,\hat\nu_2),i=2
    \end{array}\right.
\end{equation}
Since $\epsilon>0$ is arbitrary, thus, the Defender has made guesses that $(v_i^{\text{max}},v_j^{\text{max}})\notin\{(\hat\nu_i,\nu_j)|\nu_j>\mu_j(t_1)\}$ at the instant $t_1$ based on Definition 5. 

For $k\ge2$, if $\mu_j(t_k)\ge\min_{l=1,\cdots,k-1}\mu_j(t_l)$, then the Defender does not make a guess at the instant $t_k$ since the guesses $(v_i^{\text{max}},v_j^{\text{max}})\notin\{(\hat\nu_i,\nu_j)|\nu_j>\mu_j(t_k)\}$ have already been made before $t_k$. Otherwise,  $\forall  t^{\prime}\in(0,t_k),t^{\prime\prime}\in(t_k,t_k+\delta_k)$, for all $\nu_j\in[\mu_j(t_k)+\epsilon_k,\min_{l=1,\cdots,k-1}\mu_j(t_l)]$, it also has that \eqref{lemma9_proof1} holds. Therefore, the Defender has made guesses that $(v_i^{\text{max}},v_j^{\text{max}})\notin\{(\hat\nu_i,\nu_j)|\nu_j\in(\mu_j(t_k),\min_{l=1,\cdots,k-1}\mu_j(t_l)]\}$ at the instant $t_k$ based on Definition 5. This completes the proof.
\end{proof}

\begin{lemma}\label{lemma:guess2}
        Given a trajectory of game states $\mathbf{x}_{[0,T]}=\{\mathbf{x}(t)|t\in[0,T]\}$ with $T>0$ such that $\mathbf{x}\in\mathcal{C},\forall t\in[0,T]$ and the speed $h_2(\mathbf{x}(0),\mu_2(0))<\nu_1<\mu_1(0)$, there exists time $0<T_1\le T$ such that $\nu_1\ge h_2(\mathbf{x}(t),\mu_2(t)),\forall t\in[0,T_1]$, and if 
        \begin{equation}
         h_{\nu_1,\min}^{[0,T_1]}\triangleq\min_{t\in[0,T_1]}h(\mathbf{x}(t),\nu_1)<h(\mathbf{x}(0),\nu_1),\label{h_decrease}  
        \end{equation}
        then the Defender has made guesses that $(v_1^{\text{max}},v_2^{\text{max}})\notin\mathcal{S}_{h,\nu_1}(\mathbf{x}_{[0,T]})\triangleq\{(\nu_1,\nu_2)|\nu_2\in(h_{\nu_1,\min}^{[0,T_1]},h(\mathbf{x}(0),\nu_1)]\}$. 
\end{lemma}
\begin{proof}
        The existence of $T_1\in(0,T]$ satisfying $h_2(\mathbf{x}(t),\mu_2(t))\le\nu_1$ for all $t\in[0,T_1]$ follows directly from the initial condition $h_2(\mathbf{x}(0),\mu_2(0))<\nu_1<\mu_1(0)$ and the continuity of $h_2(\mathbf{x},v)$ and $\mu_i(\mathbf{x}(t))$ for $i=1,2$.

If the condition \eqref{h_decrease} holds, then for any value $\hat h\in(h_{\nu_1,\min}^{[0,T_1]},h(\mathbf{x}(0),\nu_1)]$, there exists a set of the time instants $\Lambda=\{t_{1},t_{2},\cdots,t_p\}$, where $p$ is a positive integer, such that $h(\mathbf{x}(t_H),\nu_1)=\hat h$ and $\exists\delta_H>0,h(\mathbf{x}(t^\prime),\nu_1)<h(\mathbf{x}(t_H),\nu_1)$, $\forall t^\prime\in(t_H,t_H+\delta_H), \forall t_H\in\Lambda$.
Take $t_m=\min\Lambda$, then $\exists\delta_m>0$ such that $h(\mathbf{x}(t),\nu_1)<\hat h$ for $t\in(t_{m}, t_{m}+\delta_m)$ with $\hat h=\min_{t\in[0,t_m]} h(\mathbf{x}(t),\nu_1)$. Correspondingly, there is $\forall  t^{\prime}\in(0,t_m),t^{\prime\prime}\in(t_m,t_m+\delta_m)$,
\begin{equation}
    2=J^*(\mathbf{x}(t^{\prime}), \nu_1,\hat h)>1\ge J^*(\mathbf{x}(t^{\prime\prime}),\nu_1 ,\hat h). 
\end{equation} 

Therefore, the Defender makes a guess that $(v_1^{\text{max}},v_2^{\text{max}})\neq(\nu_1,\hat h)$ at the instant $t_m$ based on Definition 5. 
This completes the proof.
\end{proof}

\subsection{The dilemma}

In this subsection, we present the main results of this letter, which address the questions posed in Section \ref{section:problem_formulation}.
To enhance clarity, we define the function $C(x,y)\triangleq \cos(\cos^{-1}(x)+\cos^{-1}(y))$.
Further, we introduce forcing states, denoted as $\mathcal{Q}\triangleq\mathcal{Q}_0\cup\mathcal{Q}_{1}\cup\mathcal{Q}_{2}$ to characterize the states in which the Defender is compelled to make guesses, where
\begin{align}
        \mathcal{Q}_0\triangleq\big\{&\mathbf{x}|\boldsymbol\eta_{D}(\mathbf{z}_1,\mu_1)\cdot\boldsymbol\eta_{D}(\mathbf{z}_2,\mu_2)<C(\tfrac{v_1^{\text{pub}}}{\mu_1},\tfrac{v_2^{\text{pub}}}{\mu_2})\big\}\\
        \mathcal{Q}_{1}\triangleq\big\{&\mathbf{x}|\exists\nu_1\in(h_2(\mu_2), \mu_1),(\nu_1,h(\nu_1))>(v_1^{\text{pub}},v_2^{\text{pub}}),\nonumber\\
        &\boldsymbol\zeta_{D}^{1,l}(\mathbf{x},\nu_1,h(\nu_1))\cdot\boldsymbol\eta_{D}(\mathbf{z}_2,\mu_2)<C(\tfrac{v_2^{\text{pub}}}{\mu_2},\sigma_1(\nu_1)),\nonumber\\
        &\forall l\in\{1,\cdots,m_1\},h_1(\nu_1)> h_2^{-1}(\nu_1)\big\},\\
        \mathcal{Q}_{2}\triangleq\big\{&\mathbf{x}|\exists\nu_1\in(h_2(\mu_2), \mu_1),(\nu_1,h(\nu_1))>(v_1^{\text{pub}},v_2^{\text{pub}}),\nonumber\\
        &\boldsymbol\zeta_{D}^{2,l}(\mathbf{x},h(\nu_1),\nu_1)\cdot\boldsymbol\eta_{D}(\mathbf{z}_1,\mu_1)
        <C(\tfrac{v_1^{\text{pub}}}{\mu_1},\sigma_2(\nu_1))\nonumber\\
        &\forall l\in\{1,\cdots,m_2\},h_1(\nu_1)< h_2^{-1}(\nu_1)\big\}
\end{align}
and $\sigma_1(\nu_1)\triangleq -\|\mathbf{C}_{\lambda,14}|_{\nu_1,h(\nu_1)}\|^{-1}\cdot\big(v_1^{\text{pub}}\boldsymbol{\eta}_{A_1}(\mathbf{z}_1,\mu_1)\cdot\mathbf{C}_{\lambda,12}|_{\nu_1,h(\nu_1)}+v_2^{\text{pub}}\boldsymbol{\eta}_{A_2}(\mathbf{z}_2,\mu_2)\cdot\mathbf{C}_{\lambda,13}|_{\nu_1,h(\nu_1)}\big)$, $\sigma_2(\nu_1)\triangleq -\|\mathbf{C}_{\lambda,24}|_{h(\nu_1),\nu_1}\|^{-1}\cdot\big(v_1^{\text{pub}}\boldsymbol{\eta}_{A_1}(\mathbf{z}_1,\mu_1)\cdot\mathbf{C}_{\lambda,22}|_{h(\nu_1),\nu_1}+v_2^{\text{pub}}\boldsymbol{\eta}_{A_2}(\mathbf{z}_2,\mu_2)\cdot\mathbf{C}_{\lambda,23}|_{h(\nu_1),\nu_1}\big)$.

The following lemma supports subsequent analysis.
\begin{lemma}[\cite{wu2026inducing}]\label{lemma:inequality_vector}
        Let $\mathbf{e}_1,\mathbf{e}_2\in\mathbb{R}^2$ be unit vectors such that $\mathbf{e}_1\cdot\mathbf{e}_2<\cos(a+b)$, where $a,b\in(0,\frac{\pi}{2})$. Then no unit vector $\mathbf{x}\in\mathbb{R}^2$ satisfies both $\mathbf{x}\cdot\mathbf{e}_1\ge\cos a$ and $\mathbf{x}\cdot\mathbf{e}_2\ge\cos b$.
\end{lemma}
\begin{theorem}\label{theorem:dilemma}
        Consider the game formulated in Section \ref{section:problem_formulation} under Assumption \ref{assumption:maximum speed}, where the initial state satisfies $\mathbf{x}(0)\in(\mathcal{C}\cap\mathcal{M}\cap\mathcal{Q})\setminus\mathcal{W}$. The information-limiting strategy \eqref{deceptive_strategy_a1}-\eqref{deceptive_strategy_a2} proposed for the Attackers creates a dilemma for the Defender. Specifically, 
        before the Defender can definitively determine its optimal capture order (i.e., when $\mathbf{x}\notin(\mathcal{C}\cap\mathcal{M})\setminus\mathcal{W}$), it is compelled to make guesses that
        the maximum speed of the Attackers $(v_1^{\text{max}},v_2^{\text{max}})$ does not belong to a non-empty subset of $[0,1]\times[0,1]\setminus (\mu_1(0),1]\times(\mu_2(0),1]$.
        Furthermore, an incorrect guess by the Defender will reduce the payoff, while correct guesses provide no additional benefit.
        Additionally, the Attackers incur no risk when implementing the proposed strategy.
\end{theorem}

\begin{proof}
        Since $\mathbf{x}(0)\in(\mathcal{C}\cap\mathcal{M})\setminus\mathcal{W}$, we have $\max\{\hat J_1(\mathbf{x}(0),v_1^{\text{max}},v_2^{\text{max}}),\hat J_2(\mathbf{x}(0),v_2^{\text{max}},v_1^{\text{max}})\}\notin[-\varepsilon,0)$, and $v_i^{\text{max}}\in[v_i^{\text{pub}},\mu_i(0)]$.
        Both Attackers are removed in finite time if uncaptured, thus $\mathbf{x}\notin(\mathcal{C}\cap\mathcal{M})\cap\mathcal{W}$ holds at finite time $t_d\triangleq\inf\{\tau\ge0|\mathbf{x}(\tau)\notin(\mathcal{C}\cap\mathcal{M})\setminus\mathcal{W}\}$.
        
        Based on \eqref{rate_mu_k_1}, 
        $\dot{\mu}_i(t)<0$ occurs when $\mathbf{u}_D(t)\cdot\boldsymbol\eta_D(\mathbf{z}_i,\mu_i)<\frac{v_i^{\text{pub}}}{\mu_i}$. By Lemma \ref{lemma:inequality_vector}, at least one inequality $\dot{\mu}_1(t)<0$ or $\dot{\mu}_2(t)<0$ holds when $\mathbf{x}(t)\in\mathcal{Q}_0$. 
        
        For $\mathbf{x}(t)\in\mathcal{Q}_1$ with $\nu_1\in(h_2(\mu_2),\mu_1)$ satisfying $(\nu_1,h(\nu_1))>(v_1^{\text{pub}},v_2^{\text{pub}})$ and $h_1(\nu_1)>h_2^{-1}(\nu_1)$, condition $\dot{h}(t^+)=\dot{h}_1(t^+)<0$ requires $\mathbf{u}_D(t)\cdot\boldsymbol\zeta^{1,l}_D(\mathbf{x},\nu_1,h(\nu_1))<\sigma_1(\nu_1)$. Thus, at least one of  $\dot{\mu}_i(t)<0$ and $\dot{h}(t^+)<0$ must hold. The same applies for $\mathcal{Q}_2$.

        Since $\mathbf{x}(0)\in\mathcal{Q}$, at least one condition from \eqref{mu_i_decrease} for $i=1,2$ or \eqref{h_decrease} is satisfied over $[0, t_d]$ regardless of the Defender's strategy.  Consequently, four cases arise for the first time at $t_d$.

        Case 1: $v_i^{\text{max}}=\mu_i(t_d)>\mu_i(t^\prime_d),v_j^{\text{max}}>\mu_j(t_d), \exists i=1,2,j=\{1,2\}\setminus i$, where $t_d^\prime\in(t_d,t_d+\delta_1)$, $\delta_{1}>0$. The Attackers lose motivation from $t_d$. Consequently, the Defender incorrectly excludes the speed pair $(v_i^{\text{max}},v_j^{\text{max}})\in\mathcal{S}_{\mu_i}(\mathbf{x}_{[0,t_d+\delta_1]})$ at $t=t_d$, reducing payoff from 1 to 0.
        
        Case 2: $\max\{\hat J_1(\mathbf{x}(t_d),v_1^{\text{max}},v_2^{\text{max}}),\hat J_2(\mathbf{x}(t_d),v_2^{\text{max}},v_1^{\text{max}})\}=0>\max\{\hat J_1(\mathbf{x}(t^\prime_d),v_1^{\text{max}},v_2^{\text{max}}),\hat J_2(\mathbf{x}(t^\prime_d),v_2^{\text{max}},v_1^{\text{max}})\}$, where $t_d^\prime\in(t_d,t_d+\delta_2)$, $\delta_{2}>0$. This indicates that $h(\mathbf{x}(t^\prime_d),v_1^{\text{max}})<v_2^{\text{max}}=h(\mathbf{x}(t_d),v_1^{\text{max}})<h(\mathbf{x}(0),v_1^{\text{max}})$. Here, the Attackers enter the warning condition \eqref{warning_condition} from $t_d$.
        Thus, the Defender incorrectly excludes the speed pair $(v_1^{\text{max}},v_2^{\text{max}})\in\mathcal{S}_{h,v_1^{\text{max}}}(\mathbf{x}_{[0,t_d+\delta_2]})$ at $t_d$, resulting in a payoff reduction from 2 to 1.
        
        Case 3: $\max\{\hat J_1(\mathbf{x}(t_d),v_1^{\text{max}},v_2^{\text{max}}),\hat J_2(\mathbf{x}(t_d),v_2^{\text{max}},v_1^{\text{max}})\}=-\varepsilon$. Thus, $\max\{\hat J_1(\mathbf{x}(0),v_1^{\text{max}},v_2^{\text{max}}),\hat J_2(\mathbf{x}(0),v_2^{\text{max}},v_1^{\text{max}})\} <- \varepsilon$, which indicates $J^*(\mathbf{x}(0),v_1^{\text{max}},v_2^{\text{max}})=1$. The Attackers are poised to enter the warning condition \eqref{warning_condition} from time $t_d$. Hence, the Defender excludes $\mathcal{E}\triangleq \mathcal{S}_{\mu_1}(\mathbf{x}_{[0,t_d]})\cup \mathcal{S}_{\mu_2}(\mathbf{x}_{[0,t_d]}) \cup \mathcal{S}_{h,\nu_1}(\mathbf{x}_{[0,t_d]}),\exists\nu_1\in(h_2(\mathbf{x}(0),\mu_2(0)),\mu_1(0))$, but makes only correct exclusions since $(v_1^{\text{max}},v_2^{\text{max}})\notin\mathcal{E}$ due to $J^*(\mathbf{x}(t),v_1^{\text{max}},v_2^{\text{max}})=1,\forall t\in[0,t_d]$.

        Case 4: $v_1^{\text{pub}}=\mu_1(t_d)$ or $v_2^{\text{pub}}=\mu_2(t_d)$. In this case, the mismatching order condition is violated at $t_d$. Consequently, the Defender excludes  $\mathcal{E}$ but makes no incorrect guesses until time $t_d$  since $(v_1^{\text{max}},v_2^{\text{max}})\notin\mathcal{E}$ due to $\mathbf{x}(t)\in\mathcal{M}\setminus\mathcal{W}, \forall t\in[0,t_d]$, maintaining the payoff.
\end{proof}

        

\section{Visualizations of Dilemma Conditions}\label{section:visualization}
This section visualizes the states where Attackers can create dilemmas, i.e. $\mathbf{x}\in(\mathcal{Q}\cap\mathcal{C}\cap\mathcal{M})\setminus\mathcal{W}$. For the 6-dimensional 2A1D game, we fix the positions of Attacker $A_1$ and Defender $\mathbf{x}_D$ as $\mathbf{x}_{A_1}=[-2,3]$ and $\mathbf{x}_D=[0,5]$, respectively, while varying Attacker $A_2$'s position. We examine two cases with identical known speeds $v_1^{\text{pub}}=v_2^{\text{pub}}=0.3$ and warning threshold $\varepsilon=0.3$ but different heterogeneous maximum speed values $v_1^{\text{max}},v_2^{\text{max}}$. Each case is analyzed numerically over the region $[-11,11]\times[0.01, 20]\subset\mathbb{R}^2$ using the grid interval of 0.021. Whether $\mathbf{x}\in\mathcal{Q}_i$ for $i=1,2$ is verified by selecting 5 critical speed pairs uniformly. Fig. \ref{fig:dilemma} presents the resulting visualizations. Both cases demonstrate that $(\mathcal{Q}\cap\mathcal{C}\cap\mathcal{M})\setminus\mathcal{W}$ is non-empty and occupies a substantial portion of the state space, confirming ample opportunities for Attackers deception.

\begin{figure}[htbp]
        \centering
        \subfloat[]{\includegraphics[width=0.22\textwidth]{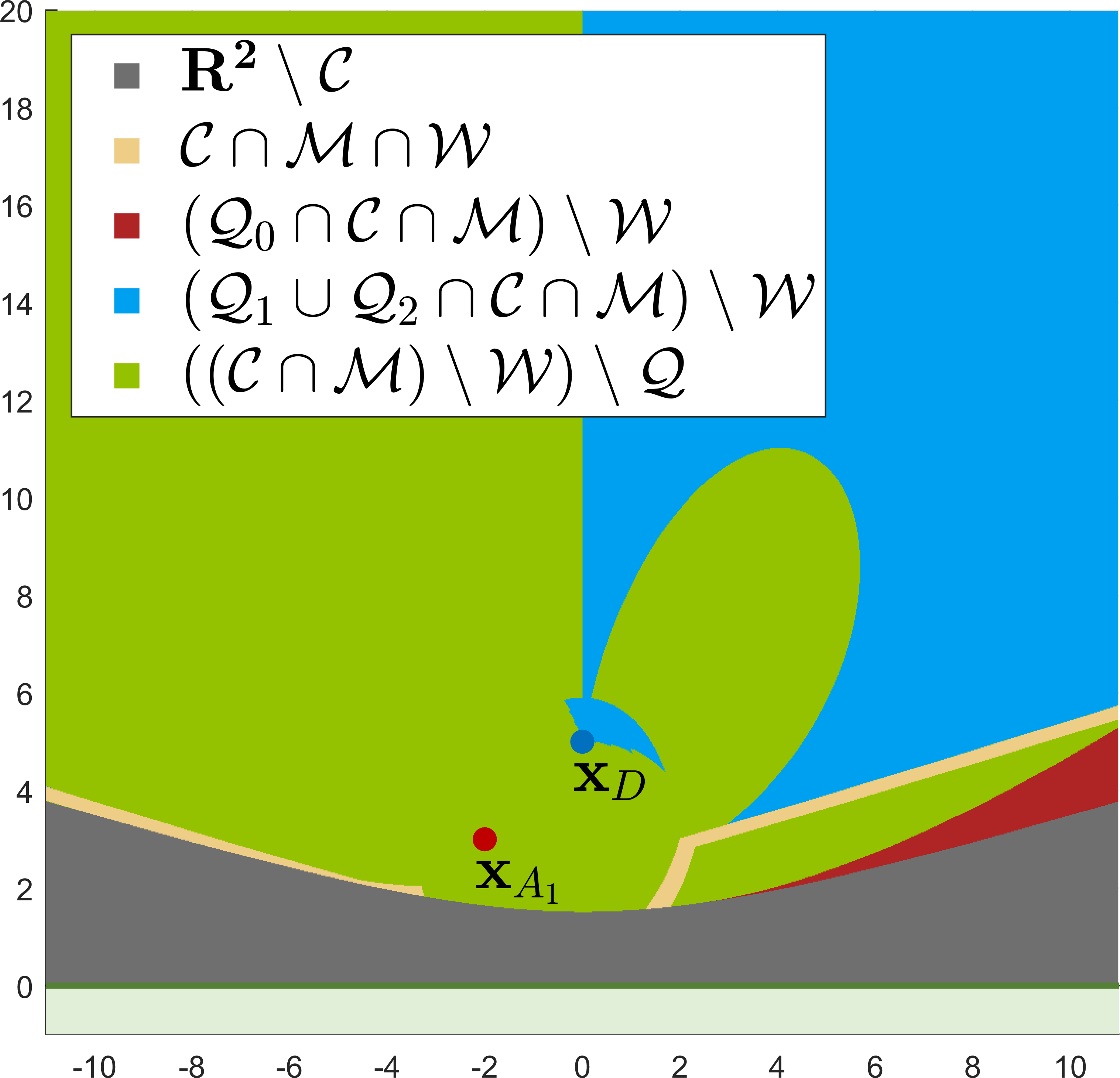}%
        \label{fig:visual1}}
        \hfill
        \subfloat[]{\includegraphics[width=0.22\textwidth]{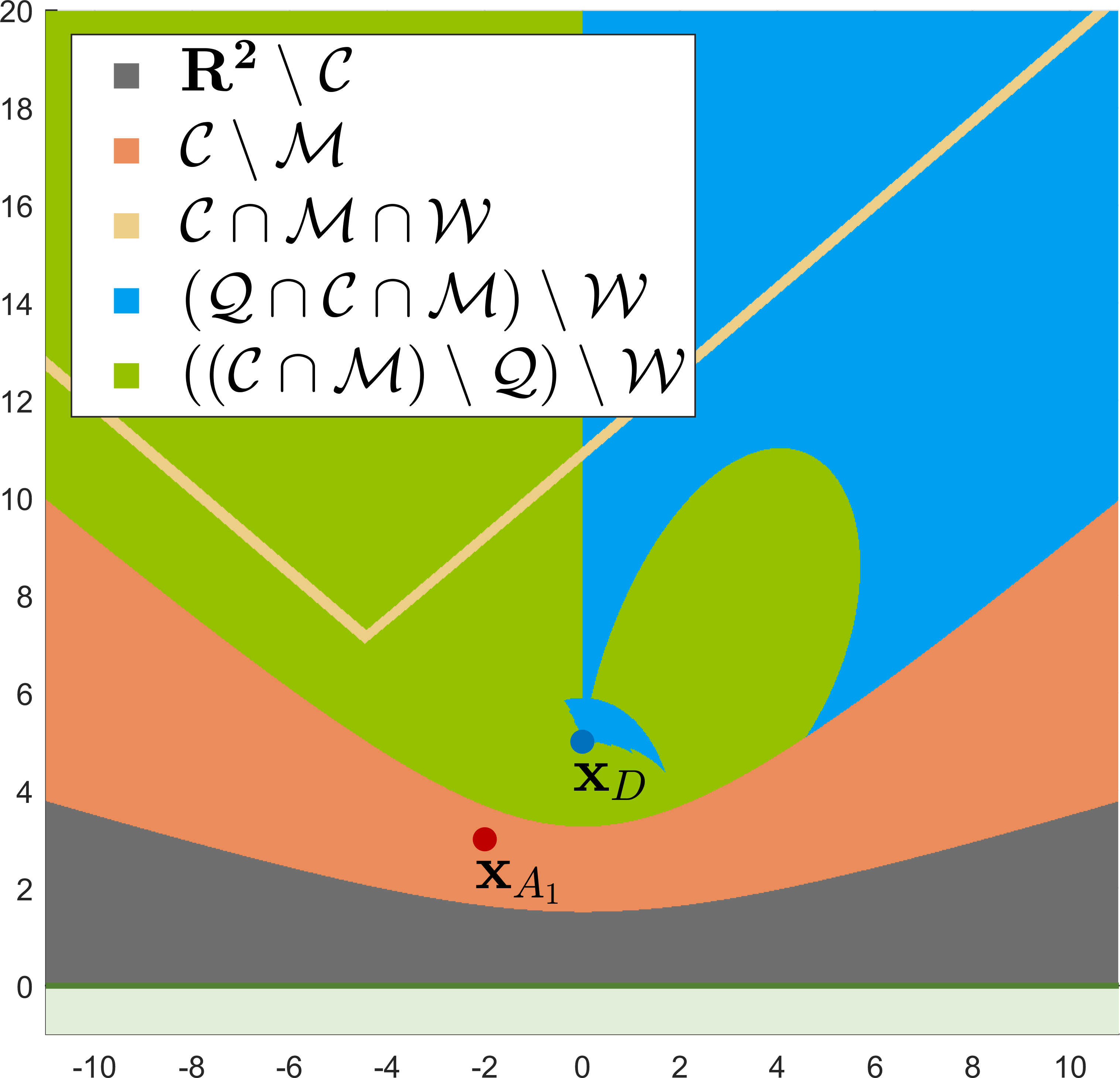}%
        \label{fig:visual2}}
        \caption{The visualizations of the dilemma conditions for the given $\mathbf{x}_{A_1}$ and $\mathbf{x}_{D}$ with varied position of $\mathbf{x}_{A_2}$. (a) Case 1 with $v_1^{\text{max}}=0.3,v_2^{\text{max}}=0.3$, (b) Case 2 with $v_1^{\text{max}}=0.61,v_2^{\text{max}}=0.65$.}
        \label{fig:dilemma}
\end{figure}

\section{Conclusion}
In this letter, we propose a critical speed pair framework to address the 2A1D reach-avoid game under incomplete information, where the Attackers' heterogeneous maximum speeds are publicly known to lie within continuous ranges.
Through examination of mismatching order states, motivation states, warning states, and forcing states, we identify the conditions under which uncertainty about the Attackers' relative capability configuration creates strategic dilemmas for the Defender.
The results demonstrate that Attackers can effectively employ the proposed slow-speed information-limiting strategy to compel the Defender to exclude certain speed pair possibilities, thereby reducing the Defender's payoff when an incorrect guess is made, without introducing additional risk to the Attackers. 
Unlike the homogeneous speed setting, heterogeneity allows the Attackers to conceal their relative capability configuration, thereby providing an additional mechanism for exploiting information advantages.
Future research directions include extending the deception analysis to multi-player scenarios in more complex environments.

\bibliographystyle{IEEEtran}
\bibliography{biblio}

\end{document}